\documentclass{article}
\usepackage{amsthm}
\newtheorem*{rep@theorem}{\rep@title}
\newcommand{\newreptheorem}[2]{%
\newenvironment{rep#1}[1]{%
 \def\rep@title{#2 \ref{##1}}%
 \begin{rep@theorem}}%
 {\end{rep@theorem}}}

\newreptheorem{theorem}{Theorem}
\newreptheorem{lemma}{Lemma}
\newreptheorem{proposition}{Proposition}
\newreptheorem{claim}{Claim}
\newreptheorem{fact}{Fact}
\usepackage[linesnumbered]{algorithm2e}

\newtheorem{theorem}{Theorem}[section]

\newtheorem{lemma}[theorem]{Lemma}

\usepackage{mathtools}
\usepackage{amsmath}
\usepackage{bbm}
\usepackage{amsfonts}

\newcommand{\E}{\mathbb{E}}
\DeclareMathOperator*{\argmax}{arg\,max}
\DeclareMathOperator*{\argmin}{arg\,min}

\renewcommand{\paragraph}[1]{\noindent\textbf{#1}.}

\title{Dynamic Averaging Load Balancing on Cycles}
\date{} 

\author{Dan Alistarh \\ IST Austria \and
        Giorgi Nadiradze \\ IST Austria \and
        Amirmojtaba Sabour \\IST Austria}

\begin{document}

\maketitle{}

\begin{abstract}
We consider the following dynamic load-balancing process: 
given an underlying graph $G$ with $n$ nodes, 
in each step $t\geq 0$, one unit of load is created, and placed at a randomly chosen graph node. 
In the same step, the chosen node picks a random neighbor, and the two nodes \emph{balance} their loads by averaging them. 
We are interested in the expected gap between the minimum and maximum loads at nodes as the process progresses, and its dependence on $n$ and on the graph structure. 

Similar variants of the above graphical balanced allocation process have been studied by Peres, Talwar, and Wieder~\cite{PTW}, and by Sauerwald and Sun~\cite{SS} for regular graphs. These authors left as open the question of characterizing the gap in the case of \emph{cycle graphs} in the \emph{dynamic} case, where weights are created during the algorithm's execution. For this case, the only known upper bound is of $\mathcal{O}( n \log n )$, following from a majorization argument due to~\cite{PTW}, which analyzes a related graphical allocation process. 

In this paper, we provide an upper bound of $\mathcal{O} ( \sqrt n \log n )$ on the expected gap of the above process for cycles of length $n$. 
We introduce a new potential analysis technique, which enables us to bound the difference in load between $k$-hop neighbors on the cycle, for any $k \leq n / 2$. 
We complement this with a ``gap covering'' argument, which bounds the maximum value of the gap by bounding its value across all possible subsets of a certain structure, and recursively bounding the gaps within each subset. 
We provide analytical and experimental evidence that our upper bound on the gap is tight up to a logarithmic factor. 
\end{abstract}

\maketitle

\section{Introduction}

We consider balls-into-bins processes where a sequence of $m$ weights are placed into $n$ bins via some randomized procedure, 
with the goal of minimizing the load imbalance between the most loaded and the least loaded bin. 
This family of processes has been used to model several practical allocation problems, such as load-balancing~\cite{ABKU, Mitz01, Richa01}, hashing~\cite{cuckoo1}, or even relaxed data structures~\cite{PODC17, SPAA18}. 

The classic formulation of this problem is known as the $d$-choice process, in each step, a new weight is generated, and is placed in the least loaded of $d$ uniform random choices. 
If $d = 1$, then we have the classic uniform random choice scheme, whose properties are now fully understood. 
In particular, if we place $m = n$ unit weights into the bins, then it is known that the most loaded bin will have expected $\Theta( \log n / \log \log n )$ load, 
whereas if $m = \Omega ( n \log n )$ we have that the expected maximum load is $m / n + \Theta( \sqrt{ m \log n / n } )$. 
Seminal work by Azar, Broder, Karlin, and Upfal~\cite{ABKU} showed that, if we place $n$ unit weights into $n$ bins by the $d$-choice process with $d \geq 2$, then, surprisingly, the 
maximum load is reduced to $\Theta( \log \log n  / \log d)$. 
A technical tour-de-force by  Berenbrink,  Czumaj,  Steger, and Vöcking ~\cite{Berenbrink00} extended this result to the ``heavily-loaded'' case where $m \gg n$, showing that in this case the maximum load is 
$m / n + \log \log n / \log d + O(1)$ with failure probability at most $1 / \textnormal{ poly } n$. An elegant alternative proof for a slightly weaker version of this result  was later provided by Talwar and Wieder~\cite{TW14}.  

More recently, Peres, Talwar, and Wieder~\cite{PTW} analyzed the \emph{graphical} version of this process, 
where the bins are the vertices of a graph, an \emph{edge} is chosen at every step, and the weight is placed at the less loaded endpoint of the edge, breaking ties arbitrarily. 
(Notice that the classic $2$-choice process corresponds to the case where the graph is a clique.) 
The authors focus on the evolution of the gap between the highest and lowest loaded bins, showing that, for $\beta$-regular expander graphs, 
this gap is $O( \log n / \beta )$, with probability $1 - 1 / \textnormal{poly } n$. 

In the \emph{static case}, where each node starts with an arbitrary initial load, and the endpoints \emph{average} their initial loads whenever the edge is chosen, the balancing process can be mapped to a Markov chain, and its convergence is well-understood in terms of the spectral gap of the underlying graph~\cite{SS}. 
Sauerwald and Sun~\cite{SS} considered this static case in the \emph{discrete} setting, where the fixed initial load can only be divided to \emph{integer} tokens upon each averaging step, for which they gave strong upper bounds for a wide range of graph families. 
By contrast, in this paper we consider the less complex \emph{continuous averaging case}, where exact averaging of the weights is possible, but in the more challenging \emph{dynamic} scenario, where weights arrive in each step rather than being initially allocated. 

One question left open by the line of previous work concerns the evolution of the gap in the dynamic case on graphs of low expansion, such as cycles. 
In particular, for cycles, the only known upper bound on the expected gap in the dynamic case is of $O( n \log n )$, following from~\cite{PTW}, 
whereas the only lower bound is the immediate $\Omega( \log n )$ gap lower bound for the clique. 
Closing this gap for cycle graphs is known to be a challenging open problem~\cite{YP}. As suggested in~\cite{PTW}, to deal with the cycle
case, there is a need for a new approach, which takes the structure of the load balancing graph into the account.

\paragraph{Contribution} In this paper, we address this question for the case where averaging is performed on a cycle graph. 
We provide an upper bound on the gap of expected $O( \sqrt{n} \log n )$ in the dynamic, heavily-loaded case, via a new potential argument. 
We complement this result with a lower bound of $\Omega( n )$ on the \emph{square} of the gap, as well as additional experimental evidence suggesting that our upper bound is tight within a logarithmic factor. 
Our results extend to \emph{weighted} input distributions. 

\paragraph{Technical Argument} 
Our upper bound result is based on two main ideas. 
The first introduces a new parametrized hop-potential function, which measures the squared difference in load between any $k$-hop neighbors on the graph, where $k \geq 1$ is a fixed hop parameter. 
That is, if $G = (V, E)$ is our input graph, and $x_i(t)$ is the load at node $i$ at time $t$, then we define the $k$-hop potential as:  
\begin{equation*} 
	\phi_k(t) = \sum_{i=1}^{n} (x_i(t) - x_{i+k}(t))^2.
\end{equation*}

The first technical step in the proof is to understand the expected (``steady-state'') value of the $k$-hop potential. 
We show that, in expectation, the $k$-hop potential has a regular recursive structure on regular graphs. 
While the expected values of $k$-hop potentials cannot be computed precisely, we can isolate upper and lower bounds on their values for cycles. In particular, for the $k$-hop potential on an $n$-cycle, we 
 prove the following bound:
\begin{equation}
\label{eq:intro-upper-bound}
    \E[\phi_k(t)] \le k(n-k)-1, \forall k \geq 1.
\end{equation}

In the second technical step, we shift gears, aiming to bound the \emph{maximum possible value} of the gap between any two nodes, leveraging the fact that we understand the hop potential for any $k \geq 1$. 
We achieve this via a ``gap covering'' technique, which characterizes the maximum value of the gap across all possible subsets of a certain type. 

More precisely, in the case of a cycle of length $n = 2^m$, for each node $i$ and hop count $k$, we define the set family $A_{k}^i$ to be formed of nodes $\{i, i+2^k, i + 2\times 2^k, i + 3\times 2^k, \dots\}$. (Since we are on a cycle, $i = i + 2^{m - k} 2^k$.) 
Then for any $1 \le i \le n$ and $k > 0$, we will have 
\begin{equation}
\sum_{i=1}^n Gap_{A_{k-1}^i}(t) \le 
\sum_{i=1}^n Gap_{A_{k}^i}(t)+\frac{n}{\sqrt{2^{k-1}}}\sqrt{\phi_{2^{k-1}}(t)},
\end{equation}

where $Gap_{X}(t)$ is the maximal gap inside the set $X$ at time $t$. 
Intuitively, this result allows us recursively characterize the gap value at various ``resolutions'' across the graph. 

Finally, we notice that we can ``cover'' the gap across between \emph{any} two nodes by carefully unwinding the recursion in the above inequality, considering all possible subsets of a well-chosen structure, and recursively bounding the gaps within each subset. 
(This step is particularly delicate in the case where $n$ is not a power of two, which we leave to the Appendix.) 
We obtain that
\begin{equation}
\E[Gap(t)] = O(\sqrt{n} \log(n)),
\end{equation}
as claimed. The logarithmic slack is caused by the second term on the right-hand-side of (2). We note that this technique extends to the case where inserted items are \emph{weighted}, where the weights are coming from some distribution of bounded second moment. 

\paragraph{Lower Bound}
It is interesting to ask whether this upper bound is tight. 
To examine this question, we revisit the recursive structure of the $k$-hop potential, which we used to obtain the lower bound in Equation~\ref{eq:intro-upper-bound}. 
We can leverage this structure to obtain a \emph{lower bound} on the expected $k$-hop potential as well. 
Starting from this lower bound, we can turn the upper bound argument ``inside out,'' to obtain a linear lower bound on the \emph{expected squared gap}: 
\begin{equation}
\E[Gap(t)^2] = \Omega(n).
\end{equation}

This second moment bound strongly suggests that our above analysis is tight within logarithmic factors. 
We conjecture that the bound is also tight with regards to the expected gap, and examine this claim empirically in Section~\ref{sec:experiments}. 

\paragraph{Extensions and Overview} 
The analysis template we described above is general, and could be extended to other graph families, such as regular expanders. 
In particular, we note that the recursive structure of the $k$-hop potentials is preserved for such graphs. The main technical steps in analyzing a new graph  family are to (1) identify the right upper bound on the $k$-hop potential (the analogue of (1)); 
and 
(2) identify the right set family for the gap covering argument, and its recursive structure (the analogue of (2)). 
Obtaining tight bounds for these quantities is not straightforward, since they do not seem to be immediately linked to well-studied graph properties. 
Here, we focus on obtaining tight bounds on the gap for cycles, which is technically non-trivial, and leave the extensions for other graph families as future work. 
To substantiate our generality claim, we exhibit an application of our analysis technique to Harary graphs~\cite{Harary} in the Appendix. 

We discuss the relation between our results and bounds for the graphical power-of-two process on a cycle~\cite{PTW} in Section~\ref{sec:discussion}.


\paragraph{Related Work}
As we have already discussed broad background, we will mainly focus on the technical differences from previous work. 
As stated, we are the first to specifically consider the \emph{dynamic} case for \emph{continuous averaging} on cycles. 
In the \emph{static} case with \emph{discrete averaging}, the problem has been considered by Sauerwald and Sun~\cite{SS}. However, their techniques would not apply in our case, since we consider that weights would be introduced \emph{dynamically}, during the processes' execution. 

To our knowledge, the only non-trivial upper bound on the gap of the process we consider which would follow from previous work is of $\mathcal{O}( n \log n )$, by the potential analysis of~\cite{PTW}: they consider 2-choice load balancing, 
and one can  re-do their potential analysis for (continuous) averaging load balancing, yielding the same bounds. However, as our bounds show, the resulting analysis is quite loose in the case of cycles, yielding an $\Omega( \sqrt n )$ gap. This is a consequence of the majorization technique used, which links dynamic averaging on the cycle and a very weak form of averaging on the clique. 

Our potential analysis is substantially different from that of~\cite{PTW}, as they track a sum of exponential potentials across the entire graph. 
By contrast, our analysis tracks the squared load differences between $k$-hop neighbors, establishing recurrences between these potentials. 
We notice that this is also different from the usual square potentials used for analyzing averaging load balancing, e.g.~\cite{Muthu}, which usually compare against the \emph{global mean}, as opposed to pairwise potential differences. Our approach is also different from the classic analyses of e.g.~\cite{ABKU}, which perform probabilistic induction on the number of bins at a given load, assuming a clique. 

Generally, our technique can be seen as performing the induction needed to bound the gap not on the bin loads, as is common in previous work, e.g.~\cite{ABKU}, but \emph{over the topology of the graph}. This approach is in some sense natural, since we wish to obtain tight, topology-specific bounds, but we believe we are the first to propose and analyze it. 

\section{Averaging on the Cycle: Upper Bounding the Gap}

\paragraph{Preliminaries} 
We consider a cycle graph $G = (V, E)$ with $n$ nodes, such that each node $i$ is connected to its left and right neighbors, $i - 1 \mod n$ and $i + 1 \mod n$.  
We consider a stochastic process following real time $t \geq 0$, in which, in each step, a weight $w(t)$ is generated from a same distribution $W$. 
We associate a real-valued \emph{load} value $x_i(t)$ with each node $i$. 
In step $t$, an edge $(i, i + 1)$ is chosen uniformly at random, and the two endpoints nodes update their weights as follows:
\[ 
	x_i(t + 1) = x_{i + 1} (t + 1) =  \frac{x_{i}(t) + x_{i+1}(t) + w(t)}{2}. \\
\]

We will assume that the second moment of the distribution $W$ is bounded. Formally, for each $t \ge 0$ there exists $M^2$ such that $E[W^2] \le M^2$. For simplicity, we will assume that
weights are normalized by $M$. This gives us that for every $t \ge 0$: $\E[W^2] \le 1$.

Let $X(t)=(x_1(t), x_2(t), ..., x_n(t))$ be the vector of the bin weights after $t$ balls have been thrown. First, we define the following potential functions: 
\[
    \forall k \in \{1, 2, \dots, n-1\}:  \phi_k(t) \coloneqq \sum_{i=1}^{n} (x_i(t) - x_{i+k}(t))^2.
\]
Above, we assume that $x_{n+i}(t)=x_i(t)$, for all
$1 \le i \le n$.
Notice that for every $1 \le i \le n$, we have that 
$\phi_i(t)=\phi_{n-i}(t)$. We want to analyze what is the value of these functions in expectation after an additional ball is thrown, for a given
load vector $X(t)$. 

We start with $\phi_1(t+1)$:
\begin{align*}
    \E[\phi_1(t&+1) | X(t),w(t)] = 
    \sum_{i=1}^{n} \frac {1}{n} \Bigg(\Big(\frac{x_{i}(t) + x_{i+1}(t) + w(t)}{2} - x_{i+2}(t)\Big)^2 \\ &\quad \quad \quad \quad \quad \quad \quad \quad + \Big(\frac{x_{i}(t) 
     + x_{i+1}(t) + w(t)}{2} - x_{i-1}(t)\Big)^2 \\ &\quad \quad \quad \quad \quad \quad \quad \quad + \sum_{j \neq i-1, i, i+1}  (x_j(t) - x_{j+1}(t))^2 \Bigg) \\ &=
     \frac{n-3}{n}\phi_1(t) + \frac{1}{2} + \frac{1}{2n}(\phi_1(t) + 2\phi_2(t)) = \frac{n-2}{n} \phi_1(t) + \frac{1}{2}(w(t)^2-\frac{\phi_1(t)}{n}) + \frac{1}{n} \phi_2(t).
\end{align*}

Now, we proceed with calculating the expected value of $\phi_k(t+1)$, for $2 \le k \le \lfloor n/2 \rfloor$:

\begin{align*}
    \E[\phi_k(t&+1) | X_t,w(t)] =
    \sum_{i=1}^{n} \frac {1}{n} \Bigg(\Big(\frac{x_{i}(t) + x_{i+1}(t) + w(t)}{2} - x_{i-k}(t)\Big)^2 
    \\ &\quad \quad \quad \quad \quad \quad \quad \quad + \Big(\frac{x_{i}(t) 
     + x_{i+1}(t) + w(t)}{2} - x_{i+1-k}(t)\Big)^2 
     \\ &\quad \quad \quad \quad \quad \quad \quad \quad +
     \Big(\frac{x_{i}(t) + x_{i+1}(t) + w(t)}{2} - x_{i+k}(t)\Big)^2 
     \\ &\quad \quad \quad \quad \quad \quad \quad \quad + \Big(\frac{x_{i}(t) 
     + x_{i+1}(t) + w(t)}{2} - x_{i+1+k}(t)\Big)^2
     \\ &\quad \quad \quad \quad \quad \quad \quad \quad + 
     \sum_{j \neq i-k, i+1-k, i+k, i+1+k}  (x_j(t) - x_{j+k}(t))^2 \Bigg) \\ &=
    \frac{n-2}{n} \phi_k(t) + (w(t)^2 - \frac{\phi_1(t)}{n}) +  \frac{\phi_{k+1}(t)}{n} + \frac{\phi_{k-1}(t)}{n}.
\end{align*}
Note that in the above calculations for $\phi_1(t+1)$ and $\phi_k(t+1)$, for $k>1$
the terms which contain $w(t)$ as linear multiplicative term
disappear because we can assume that loads $x_1(t), x_2(t), ..., x_n(t)$
are normalized (this will not change our potentials) and we have:
\begin{equation}
    \sum_{i=1}^n w(t)x_i(t)=0.
\end{equation}
If we remove conditioning on $w(t)$ and express these equations for $k = 1, 2, \dots, n-1$, we get:

\[
\begin{cases}
\E[\phi_1(t+1)|X(t)] = (\frac{n-2}{n})\phi_1(t) + \frac{1}{2}(\E[W^2] - \frac{\phi_1(t)}{n}) + \frac{\phi_2(t)}{n}. \\
\E[\phi_2(t+1)|X(t)] = (\frac{n-2}{n})\phi_2(t) + (\E[W^2] - \frac{\phi_1(t)}{n}) + \frac{\phi_1(t)}{n} + \frac{\phi_3(t)}{n}. \\
\dots \\
\E[\phi_{\lfloor \frac{n}{2} \rfloor}(t+1)|X(t)] = (\frac{n-2}{n})\phi_{\lfloor \frac{n}{2} \rfloor}(t) + (\E[W^2] - \frac{\phi_1(t)}{n})  \\
\quad \quad \quad \quad \quad \quad \quad \quad \quad +
\frac{\phi_{\lfloor \frac{n}{2} \rfloor - 1}(t)}{n} + \frac{\phi_{\lfloor \frac{n}{2} \rfloor + 1}(t)}{n}. \\
\dots \\
\E[\phi_{n-2}(t+1)|X(t)] = (\frac{n-2}{n})\phi_{n-2}(t) \\
\quad \quad \quad \quad \quad \quad \quad \quad \quad+ (\E[W^2] - \frac{\phi_1(t)}{n}) + \frac{\phi_{n-3}(t)}{n} +
\frac{\phi_{n-1}(t)}{n}. \\
\E[\phi_{n-1}(t+1)|X(t)] = (\frac{n-2}{n})\phi_{n-1}(t) + \frac{1}{2}(\E[W^2] - \frac{\phi_1(t)}{n}) + \frac{\phi_{n-2}(t)}{n}. \\

\end{cases}
\]

Using the above equations we can prove the following:
\begin{lemma} \label{thm:UpperBoundOnPotentials}
For every $t \ge 0$ and $1 \le k \le n-1$, we have that
\begin{equation}
    \E[\phi_k(t)] \le (k(n-k)-1) \E[W^2] \le k(n-k)-1.
\end{equation}
\end{lemma}

\begin{proof}
Let $\Phi(t)=(\phi_1(t), \phi_2(t), ..., \phi_{n-1}(t))$
be the vector of values of our potentials at time step $t$
and let $Y=(y_1, y_2, ..., y_{n-1})$,
be the vector containing our desired upper bounds for each potential. That is: for each $1 \le i \le n-1$,
we have that $y_i=(i(n-i)-1)\E[W^2]$.
An interesting and easily checkable thing about the vector $Y$ is that 
\begin{equation}
\E[\Phi(t+1)|\Phi(t)=Y]=Y.
\end{equation}

Next, consider the vector $Z(t)=(z_1(t), z_2(t), ... z_{n-1}(t))=Y-\Phi(t)$.
Our goal is to show that for every step $t$ and coordinate $i$,
$\E[z_i(t)]\ge 0$. we have that 
\begin{align*}
\E[z_1(t&+1)|X(t)]=y_1-\E[\phi_1(t+1)|X(t)] \\ &= 
(\frac{n-2}{n})y_1 + \frac{1}{2}(\E[W^2] - \frac{y_1}{n}) + \frac{y_2}{n} 
- \Bigg((\frac{n-2}{n})\phi_1(t) + \frac{1}{2}(\E[W^2] - \frac{\phi_1(t)}{n}) + \frac{\phi_2(t)}{n} \Bigg) \\ &=
(\frac{n-2}{n})z_1(t) - \frac{z_1(t)}{2n} + \frac{z_2(t)}{n}.
\end{align*}
and for $2 \le i \le \lfloor \frac{n}{2} \rfloor$, we have that
\begin{align*}
\E[z_i(t&+1)|X(t)]= 
(\frac{n-2}{n})z_i(t) - \frac{z_1(t)}{n} + \frac{z_{i-1}(t)}{n}+
\frac{z_{i-1}(t)}{n}.
\end{align*}

Hence we get the following equations(recall that $z_i(t)=z_{n-i}(t)$):
\begin{equation} \label{case:Z_equations}
\begin{cases} 
n \times \E[z_1(t+1)|X(t)] = (n - 2 - \frac{1}{2}) z_1(t) + z_2(t). \\
n \times \E[z_2(t+1)|X(t)] = -z_1(t) +  z_1(t) + (n - 2) z_2(t) + z_3(t). \\
n \times \E[z_3(t+1)|X(t)] = -z_1(t) +  z_2(t) + (n - 2) z_3(t) + z_4(t).\\
\dots \\
n \times \E[z_{\lfloor \frac{n}{2}\rfloor}(t+1)|X(t)] = -z_1(t) + z_{\lfloor \frac{n}{2}\rfloor-1}(t)   + 
(n-2) z_{\lfloor \frac{n}{2}\rfloor}(t) + z_{\lfloor \frac{n}{2}\rfloor+1}(t).
\end{cases}
\end{equation}
Next, using induction on $t$, we show that for every $t \ge 0$

\begin{equation} \label{eq:inductionassumption}
0 \le \E[z_1(t)] \le \E[z_2(t)] \le ... \le \E[z_{\lfloor \frac{n}{2}\rfloor}(t)].
\end{equation}
The base case holds trivially since $Z(0)=Y$.
For the induction step, assume that  $0 \le \E[z_1(t)] \le \E[z_2(t)] \le ... \le \E[z_{\lfloor \frac{n}{2}\rfloor}(t)]$.
First, we have that 
\begin{equation*}
n\E[z_1(t+1)]=n\E_{X(t)}[\E[z_1(t+1)|X(t)]]=(n - 2 - \frac{1}{2}) \E[z_1(t)] + \E[z_2(t)] \ge 0.
\end{equation*}
Additionally, we have that:
\begin{align*}
n\E[z_1(t+1)] &=(n - 2 - \frac{1}{2}) \E[z_1(t)] + \E[z_2(t)] \le (n - 2) \E[z_1(t)] + \E[z_2(t)] \\ &\le  (n - 2) \E[z_2(t)] + \E[z_3(t)]=n\E[z_2(t+1)].
\end{align*}
For $2 \le i \le \lfloor \frac{n}{2}\rfloor-2$, we have that
\begin{align*}
n\E[z_i(t+1)] &=-\E[z_1(t)]+\E[z_{i-1}(t)]+(n - 2) \E[z_i(t)] + \E[z_{i+1}(t)] \\&\le -\E[z_1(t)]+\E[z_{i}(t)]+(n - 2) \E[z_{i+1}(t)] + \E[z_{i+2}(t)]  \\&= n\E[z_{i+1}(t+1)].
\end{align*}
Next, observe that by our assumption: \\ $\E[z_{\lfloor \frac{n}{2}\rfloor+1}(t)]= \E[z_{\lceil \frac{n}{2}\rceil-1}(t)]\ge \E[z_{\lfloor \frac{n}{2}\rfloor-2}(t)]$.
Finally, by using this observation we get that
\begin{align*}
n&\E[z_{\lfloor \frac{n}{2}\rfloor-1}(t+1)] =-\E[z_1(t)]+\E[z_{\lfloor \frac{n}{2}\rfloor-2}(t)]+(n - 2) \E[z_{\lfloor \frac{n}{2}\rfloor-1}(t)] + \E[z_{\lfloor \frac{n}{2}\rfloor}(t)] \\&\le -\E[z_1(t)]+\E[z_{\lfloor \frac{n}{2}\rfloor+1}(t)]+\E[z_{\lfloor \frac{n}{2}\rfloor-1}(t)] + (n - 3) \E[z_{\lfloor \frac{n}{2}\rfloor-1}(t)]+\E[z_{\lfloor \frac{n}{2}\rfloor}(t)]
\\ &\le -\E[z_1(t)]+\E[z_{\lfloor \frac{n}{2}\rfloor+1}(t)]+\E[z_{\lfloor \frac{n}{2}\rfloor-1}(t)]+(n - 2)\E[z_{\lfloor \frac{n}{2}\rfloor}(t)] \\ &= n\E[z_{\lfloor \frac{n}{2}\rfloor}(t+1)].
\end{align*}
This completes the proof of the theorem.
\end{proof}

\section{Upper Bound on the Gap for $n=2^m$}
In this section we upper bound a gap in expectation for $n=2^m$ case. The proof for the general case
is quite technical but not necessarily more interesting, and is provided in the section \ref{sec:generalUB} in the Appendix.

We begin with some definitions. 
For a set $A \subseteq \{1, 2, \dots, n\}$, let
\[
    Gap_A(t) = \max_{i \in A} (x_i(t)) - \min_{i \in A}(x_i(t)).
\]
Also, let $A_{k}^i$ be $\{i, i+2^k, i + 2\times 2^k, i + 3\times 2^k, \dots\}$
(Notice that $i=i+2^{m-k}2^k$).  
Our proof works as follows: for each $1 \le i \le n$ and
$0 < k \le m$, we look at the vertices given by the sets $A_k^i$ and $A_k^{i+2^{k-1}}$ 
and try to characterise the gap after we merge those sets (Note that this will give us 
the gap for the set $A_{k-1}^i=A_k^i \cup A_k^{i+2^{k-1}}$). Using this result, we are able to show that 
$\sum_{i=1}^n Gap_{A^i_{k-1}}(t)$ is upper bounded by $\sum_{i=1}^n Gap_{A_{k}^i}(t)$ plus $n$ times
maximum load difference between vertices at hop distance $2^{k-1}$. 
Next, we use $2^{k-1}$ hop distance potential  $\phi_{2^{k-1}}(t)$ to upper bound
maximum load between the vertices at hop distance $2^{k-1}$.
Using induction on $k$, we are able to upper bound $\sum_{i=1}^n Gap_{A^i_{0}}(t)$ in terms of 
$\sum_{i=1}^n Gap_{A^i_{m}}(t)$ and $\sum_{k=1}^m \phi_{2^{k-1}}(t)$.
Notice that by our definitions, for each $i$, $Gap_{A^i_{m}}(t)=0$ ($A_i^m$ contains only vertex $i$) and $Gap_{A^i_{0}}(t)=Gap(t)$ ($A^{i}_0$ contains all vertices).
Hence, what is left is to use the upper bounds for the hop distance potentials, which
we derived in the previous section.

We start by proving the following useful lemma.

\begin{lemma} \label{lem:UnionOfSets}
For any $1 \le i \le n$ and $k > 0$, we have that
\begin{equation}
2Gap_{A_{k-1}^i}(t) \le 2 \max_{j \in A_{k-1}^i}|x_j(t) - x_{j+2^{k-1}}(t)| + Gap_{A_{k}^{i+2^{k-1}}}(t)+
Gap_{A_{k}^i}(t).
\end{equation}
\end{lemma}

\begin{proof}
Fix vertex $i$.
Note that $A_{k-1}^i=A_{k}^i \cup A_{k}^{i+2^{k-1}}$.
Let $u=\argmax_{j \in A_{k-1}^i} x_j(t)$ and let $v =\argmin_{j \in A_{k-1}^i} x_j(t)$.
We consider several cases on the membership of nodes $u$ and $v$, and bound the gap in each one:

\textbf{Case 1}. $u \in A_{k}^{i}$ and $v \in A_{k}^{i}$.
Then $Gap_{A_{k-1}^i}(t)=Gap_{A_{k}^i}(t)$  
and we have that

\begin{align*}
    Gap_{A_{k}^i}(t)  &= |x_u(t) - x_v(t)| \\ &\leq |x_{u+2^{k-1}}(t) - x_{u}(t)| + |x_{v+2^{k-1}}(t) - x_{v}(t)| + |x_{u+2^{k-1}}(t) - x_{v+2^{k-1}}(t)| \\
    & \leq |x_{u+2^{k-1}}(t) - x_{u}(t)| + |x_{v+2^{k-1}}(t) - x_{v}(t)| + Gap_{A_{k}^{i+2^{k-1}}}(t) \\ 
    & \leq  2 \max_{j \in A_{k-1}^i}|x_j(t) - x_{j+2^{k-1}}(t)| + Gap_{A_{k}^{i+2^{k-1}}}(t).
\end{align*}
Where we used the fact that both $u+2^{k-1}$ and $v+2^{k-1}$ belong to $A_{k}^{i+2^{k-1}}$.
This gives us that
\begin{equation}
2Gap_{A_{k-1}^i}(t) \le 2 \max_{j \in A_{k-1}^i}|x_j(t) - x_{j+2^{k-1}}(t)| + Gap_{A_{k}^{i+2^{k-1}}}(t)+
Gap_{A_{k}^i}(t).
\end{equation}

\textbf{Case 2}. $u \in A_{k}^{i}$ and $v \in A_{k}^{i+2^{k-1}}$.
Then we have that:

\begin{align*}
    Gap_{A_{k-1}^i}(t) &= |x_u(t) - x_v(t)| \leq |x_u(t) - x_{v + 2^{k-1}}(t)| + |x_{v+2^{k-1}}(t) - x_v(t)| \\
    & \leq Gap_{A_{k}^{i}}(t) +  \max_{j \in A_{k-1}^i}(|x_j(t) - x_{j+2^{k-1}}(t)|)
\end{align*}
and
\begin{align*}
    Gap_{A_{k-1}^i}(t) &= |x_u(t) - x_v(t)| \leq |x_u(t) - x_{u + 2^{k-1}}(t)| + |x_{u+2^{k-1}}(t) - x_v(t)| \\
    & \leq Gap_{A_{k}^{i+2^{k-1}}}(t) +  \max_{j \in A_{k-1}^i}(|x_j(t) - x_{j+2^{k-1}}(t)|)
\end{align*}
Where we used $v+2^{k-1} \in A_i^k$ and $u+2^{k-1} \in A_{k}^{i+2^{k-1}}$.
Hence, we again get that
\begin{equation}
2Gap_{A_{k-1}^i}(t) \le 2 \max_{j \in A_{k-1}^i}|x_j(t) - x_{j+2^{k-1}}(t)| + Gap_{A_{k}^{i+2^{k-1}}}(t)+
Gap_{A_{k}^i}(t).
\end{equation}

\textbf{Case 3}. $u \in A_{k}^{i+2^{k-1}}$ and $v \in A_{k}^{i+2^{k-1}}$, is similar to Case 1.

\textbf{Case 4}. $v \in A_{k}^{i}$ and $u \in A_{k}^{i+2^{k-1}}$, is similar to Case 2.
\end{proof}

Next, we upper bound the quantity $\sum_{i=1}^n \max_{j \in A_{k}^i}|x_j(t) - x_{j+2^{k}}(t)|.$

\begin{lemma} \label{lem:SumOfEdges}
\begin{equation}
\sum_{i=1}^n \max_{j \in A_{k}^i}|x_j(t) - x_{j+2^{k}}(t)| \le \frac{n}{\sqrt{2^k}} \sqrt{\phi_{2^k}(t)}.
\end{equation}
\end{lemma}
\begin{proof}
Notice that for any $i$ and $i' \in A_k^i$, we have that $A_k^i=A_k^{i'}$,
hence $\max_{j \in A_{k}^i}|x_j(t) - x_{j+2^{k}}(t)| = \max_{j \in A_{k}^i}|x_j(t) - x_{j+2^{k}}(t)|$
and this means that 
\begin{align*}
\sum_{i=1}^n \max_{j \in A_{k}^i}|x_j(t) &- x_{j+2^{k}}(t)| = \frac{n}{2^k} \sum_{i=1}^{2^{k}}
\max_{j \in A_{k}^{i}}|x_j(t) - x_{j+2^{k}}(t)| \\ &\le 
\frac{n}{2^k}\sqrt{2^k} \sqrt{\sum_{i=1}^{2^{k}}
\max_{j \in A_{k}^{i}}|x_j(t) - x_{j+2^{k}}(t)|^2} \\&\le
\frac{n}{2^k}\sqrt{2^k} \sqrt{\sum_{j=1}^n |x_j(t) - x_{j+2^{k}}(t)|^2}=\frac{n}{\sqrt{2^k}} \sqrt{\phi_{2^k}(t)}
\end{align*}
Where we used a fact that sets $A_k^1, A_k^2, ..., A_k^{2^k}$ are disjoint.
\end{proof}

Finally, using the two Lemmas above and Theorem \ref{thm:UpperBoundOnPotentials} we can upper bound the expected gap at step $t$:
\begin{theorem} \label{thm:GapUpperBound}
For every $t \ge 0$, we have that
\begin{equation*}
\E[Gap(t)] = O(\sqrt{n} \log(n)).
\end{equation*}
\end{theorem}
 
 \begin{proof}

From Lemma \ref{lem:UnionOfSets} we have that 
\begin{align*}
\sum_{i=1}^n 2Gap_{A_{k-1}^i(t)} &\le \sum_{i=1}^n Gap_{A_{k}^i}(t)+\sum_{i=1}^n Gap_{A_{k}^{i+2^{k-1}}}(t)  \\& \quad \quad \quad \quad \quad \quad \quad +
\sum_{i=1}^n 2 \max_{j \in A_{k-1}^i}|x_j(t) - x_{j+2^{k-1}}(t)| \\&=
2\sum_{i=1}^n Gap_{A_{k}^i}(t)+2\sum_{i=1}^n \max_{j \in A_{k-1}^i}|x_j(t) - x_{j+2^{k-1}}(t)|.
\end{align*}
After dividing the above inequality by 2 and applying Lemma \ref{lem:SumOfEdges}
we get that:
\begin{align*}
\sum_{i=1}^n Gap_{A_{k-1}^i}(t) \le 
\sum_{i=1}^n Gap_{A_{k}^i}(t)+\frac{n}{\sqrt{2^{k-1}}}\sqrt{\phi_{2^{k-1}}(t)}.
\end{align*}
Notice that $\sum_{i=1}^n Gap_0^i(t)=n Gap(t)$
and we also have that 
\begin{align*}
\sum_{i=1}^n Gap_{\frac{n}{2}}^i(t)&= \sum_{i=1}^n |x_i(t)-x_{i+\frac{n}{2}}(t)|  \le
\sqrt{n} \sqrt{\sum_{i=1}^n |x_i(t)-x_{i+\frac{n}{2}}(t)|^2}=\sqrt{n}\sqrt{\phi_{\frac{n}{2}}(t)}
\end{align*}
Hence, we get that 
\begin{align*}
n Gap(t)=\sum_{i=1}^n Gap_0^i(t) &\le \sum_{i=1}^n Gap_{\frac{n}{2}}^i(t)+\sum_{k=1}^{m-1}\frac{n}{\sqrt{2^{k-1}}}\sqrt{\phi_{2^{k-1}}(t)} \\& \le
\sqrt{n}\sqrt{\phi_{\frac{n}{2}}(t)}+\sum_{k=1}^{m-1}\frac{n}{\sqrt{2^{k-1}}}\sqrt{\phi_{2^{k-1}}(t)}.
\end{align*}
Next, we apply Jensen and Theorem \ref{thm:UpperBoundOnPotentials}:
\begin{align*}
n \E[Gap(t)]  &\le
\sqrt{n}\E\sqrt{\phi_{\frac{n}{2}}(t)}+\sum_{k=1}^{m-1}\frac{n}{\sqrt{2^{k-1}}}\E\sqrt{\phi_{2^{k-1}}(t)}
\\ &\le \sqrt{n}\sqrt{\E[\phi_{\frac{n}{2}}(t)]}+\sum_{k=1}^{m-1}\frac{n}{\sqrt{2^{k-1}}}\sqrt{\E[\phi_{2^{k-1}}(t)]} \\ &\le \sqrt{n}\sqrt{\Big(\frac{n}{2}\Big)^2}+\sum_{k=1}^{m-1}\frac{n}{\sqrt{2^{k-1}}}
\sqrt{2^{k-1}(n-2^{k-1})} \\ &\le m n\sqrt{n}=n(\log{n})\sqrt{n}.
\end{align*}
This gives us the proof of the theorem.
\end{proof}

\section{Gap Lower Bound}

Next we prove the following theorem, which provides strong evidence that our bound on the gap is tight within a logarithmic factor. 

\begin{theorem}
The following limit holds: 
\begin{equation*}
 \lim_{t\to\infty} \E[Gap(t)^2]=\Omega(n\E[W^2])).   
\end{equation*}

\end{theorem}

\begin{proof}
In this case  we want to prove that not only does vector $Z(t)$ have positive coordinates in expectation, but also $\E[z_{\lfloor \frac{n}{2}\rfloor}]$  converges to 0
. This will give us that $\phi_{\lfloor \frac{n}{2}\rfloor}$
approaches it's upper bound $(\lfloor \frac{n}{2}\rfloor 
\lceil \frac{n}{2}\rceil-1)\E[W^2]$ in expectation.
Then, we can show that there exist two nodes(At distance $\lfloor \frac{n}{2}\rfloor$) such that the expected square of difference between their loads is $\Omega(n\E[w^2])$.

Recall from Equations \ref{case:Z_equations} that
\begin{align*}
    n\E[z_{\lfloor \frac{n}{2} \rfloor (t+1)}] &= 
    -\E[z_1(t)]+\E[z_{\lfloor \frac{n}{2}\rfloor+1}(t)] +\E[z_{\lfloor \frac{n}{2}\rfloor-1}(t)]+(n - 2)\E[z_{\lfloor \frac{n}{2}\rfloor}(t)].
\end{align*}
We also know that Inequalities \ref{eq:inductionassumption} hold
for every $t$, hence we get that
\begin{equation*}
 \E[z_{\lfloor \frac{n}{2} \rfloor (t+1)}] \le
 \E[z_{\lfloor \frac{n}{2} \rfloor (t)}]-\frac{\E[z_1(t)]}{n}.
\end{equation*}
The above inequality in combination with Inequalities \ref{eq:inductionassumption} means that
\begin{align} \label{eq:UpperBound}
\E[z_{\lfloor \frac{n}{2} \rfloor}(t+\lfloor \frac{n}{2} \rfloor+1)] &\le \E[z_{\lfloor \frac{n}{2} \rfloor}(t+1)]-\sum_{i=t}^{t+\lfloor \frac{n}{2} \rfloor} \frac{\E[z_1(i)]}{n} \le \E[z_{\lfloor \frac{n}{2} \rfloor}(t+1)]- \frac{\E[z_1(t+\lfloor \frac{n}{2} \rfloor)]}{n}
\end{align}

Again by using Equations \ref{case:Z_equations} and Inequalities \ref{eq:inductionassumption}, we can show that for every $1 \le i \le \lfloor \frac{n}{2} \rfloor -1$:

\begin{equation*}
 \E[z_i(t+1)] \ge \frac{\E[z_{i+1}(t)]}{n}.
\end{equation*}
This gives us that:
\begin{align*} 
    \E[z_1(t+\lfloor \frac{n}{2} \rfloor)] &\geq 
    \Big( \frac{1}{n} \Big) \E[z_2(t+\lfloor \frac{n}{2} \rfloor-1)]
    \geq \Big( \frac{1}{n} \Big)^2 \E[z_3(t+\lfloor \frac{n}{2} \rfloor-2)] \\ & \geq \dots \\ 
    &\geq \Big( \frac{1}{n} \Big)^{\lfloor \frac{n}{2}\rfloor - 1} \E[z_{\lfloor \frac{n}{2}\rfloor}(t+\lfloor \frac{n}{2}\rfloor - (\lfloor \frac{n}{2}\rfloor - 1))] = 
    \Big( \frac{1}{n} \Big)^{\lfloor \frac{n}{2}\rfloor - 1}\E[z_{\lfloor \frac{n}{2}\rfloor}(t+1)].
\end{align*}
By plugging the above inequality in inequality \ref{eq:UpperBound}
we get that
\begin{align*}
\E[z_{\lfloor \frac{n}{2} \rfloor}&(t+\lfloor \frac{n}{2} \rfloor+1)] \le \E[z_{\lfloor \frac{n}{2} \rfloor}(t+1)]- \frac{\E[z_1(t+\lfloor \frac{n}{2} \rfloor)]}{n}\\ &\le
\E[z_{\lfloor \frac{n}{2} \rfloor}(t+1)]
-
\Big( \frac{1}{n} \Big)^{\lfloor \frac{n}{2}\rfloor - 1}\E[z_{\lfloor \frac{n}{2}\rfloor}(t+1)] = \Bigg(1-\Big(\frac{1}{n} \Big)^{\lfloor \frac{n}{2}\rfloor - 1} \Bigg)\E[z_{\lfloor \frac{n}{2}\rfloor}(t+1)]
\end{align*}
Because $\Bigg(1-\Big(\frac{1}{n} \Big)^{\lfloor \frac{n}{2}\rfloor - 1} \Bigg) < 1$ and does not depend on $t$,
we get that
$
    \lim_{t\to\infty}\E[z_{\lfloor \frac{n}{2}\rfloor}(t)]=0.
$. This means that
$
    \lim_{t\to\infty}\E[\phi_{\lfloor \frac{n}{2}\rfloor}(t)]=\Omega(n^2\E[W^2]).
$ \\ Let $Gap_{\lfloor \frac{n}{2}\rfloor}(t)=\max_{1\le i \le n}|x_i(t)-x_{i+\lfloor \frac{n}{2}\rfloor}(t)|$.
Note that: \\
$
Gap(t)^2\ge Gap_{\lfloor \frac{n}{2}\rfloor}(t)^2 \ge \frac{\phi_{\lfloor \frac{n}{2}\rfloor}(t)}{n}.
$
Hence
$
 \lim_{t\to\infty} \E[Gap(t)^2]=\Omega(n\E[W^2]).   
$

Unfortunately we are not able to obtain the lower bound on the gap, since our approach uses the fact
that the upper bounds on $k$-hop potentials are 'tight'. Since our potentials are quadratic, 
we are not able to derive any kind of lower on for the gap itself. Intuitively, this will be an issue
with any argument which uses convex potential.

\end{proof}


\section{Experimental Validation}
\label{sec:experiments}

On the practical side, we implemented our load balancing algorithm with unit weight increments on a cycle.
The results confirm our hypothesis that the gap
is of order $\Theta(\sqrt{n})$.
In Figure \ref{fig:firstfigure} we ran our experiment 100 times
and calculated average gap over the all runs.
$x$-axis shows number of balls thrown(which is the same as the number of increments) and $y$-axis is current average gap divided by $\sqrt{n}$.
The experiment shows that once the number of thrown balls is large enough,
the gap stays between $\sqrt{n}$ and $1.4\sqrt{n}$.

\begin{figure*}[ht!]%
\centering
\caption{The evolution of average gap divided by square root of $n$, where $n$ is the number of bins.}%
\label{fig:firstfigure}%
\includegraphics[width=\textwidth]{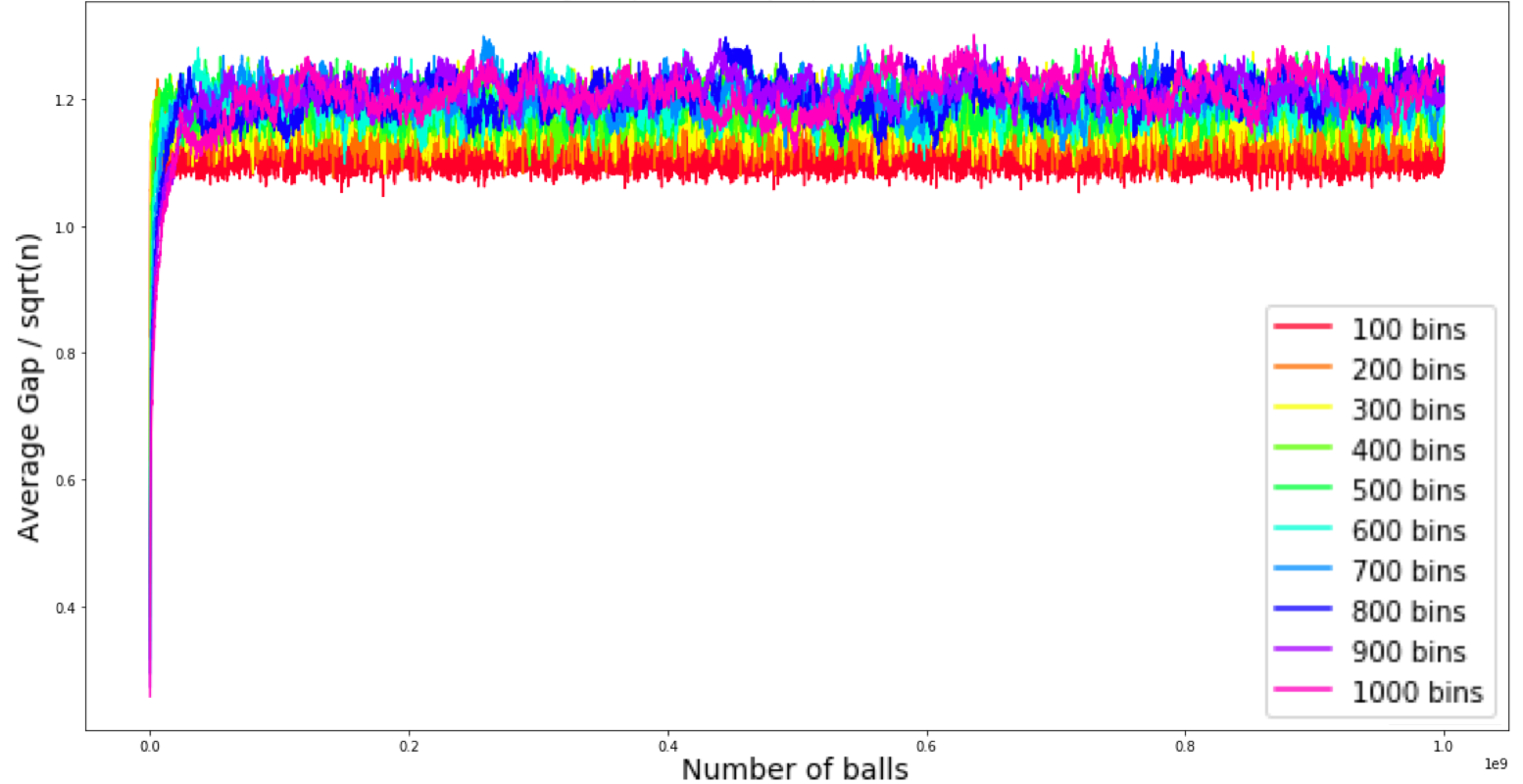}
\end{figure*}

\section{Discussion and Future Work}
\label{sec:discussion}

We have shown that in the case of dynamic averaging on a cycle
the gap between highest and lowest loaded bins is upper bounded by $O(\sqrt{n}\log{n})$ in expectation. Additionally we showed that 
the expected square of the gap is lower bounded by $\Omega(n)$.
It the future, it would be interesting to further tighten our results, matching our experimental analysis. 
We conjecture that the ``correct'' bound on the expected gap is of $\Theta(\sqrt{n})$. 
As already discussed, we also plan to extend our results to more general graph families, in particular grids graphs. 

\paragraph{Comparison of two-choice and averaging load balancing}
Finally, it is interesting to ask if it possible to extend our gap bounds in the case of the classic two-choice load balancing process. 
In particular, it is possible to show that the gap in the case of averaging process is always smaller
in expectation than the gap in the case of two choice process? 
Intuitively this should be the case, since the load balancing operation in the case of averaging can be 
viewed as picking up the random edge, incrementing the load of the endpoint with the smaller load 
and then  averaging the values. The extra averaging step should not make the gap larger. 
Indeed, the exponential potential used to analyse the gap in \cite{PTW}
can be used to upper bound the gap for averaging, since the exponential function is convex and averaging values 
does not increase it (this follows by Jensen's inequality). 

Unfortunately, it is not clear if averaging helps to actually \emph{decrease} the exponential potential. 
Additionally, this argument shows that averaging does not make the gap worse if applied to the particular technique of upper bounding the gap, and it is not clear if the gap itself is actually smaller, if we use averaging on top of two choice. We conjecture that there exists a majorization argument
which is based on \emph{how often} the process performs the averaging step. 
More precisely, we consider the setting where after the increment step (using two choice), we perform averaging with probability $\beta$.
The gap should decrease in expectation as we increase $\beta$. 
Note that the only result which lower bounds the gap for the two-choice on the cycle is the straightforward $\Omega(\log n)$ lower bound which can be shown for the clique~\cite{PTW}; 
so what makes  the existence of the majorization argument interesting is that it would allow us to show that the lower bound we derived on the second moment of the gap while always performing averaging step on the cycle ($\beta=1$) can be automatically used as the lower bound on the gap for two choice on the cycle ($\beta=0$). 
We plan to investigate this connection in future work. 

\section{Acknowledgments}

The authors sincerely thank Thomas Sauerwald for insightful discussions, and Mohsen Ghaffari, Yuval Peres, and Udi Wieder for feedback on earlier versions of this draft.

\appendix

\section{Upper Bound on the Gap, General Case} \label{sec:generalUB}
To prove the Theorem \ref{thm:GapUpperBound} for the general case, we need to redefine our sets $A_i^{k}$.
In order to do this, for each $k$ we define $2^k$ dimensional vector $\Delta_k=( \delta_k^1, \delta_k^2, ..., \delta_k^{2^k} )$.
For $k=0$, we have that $\Delta_k=(n)$.
For $\lfloor \log{n} \rfloor \ge k >0$ we set $\Delta_k=(\alpha_k, \delta_{k-1}^1-\alpha_k, \alpha_k, \delta_{k-1}^2-\alpha_k, ..., \alpha_k, \delta_{k-1}^{2^{k-1}}-\alpha_k)$.

Where, 
$$
\alpha_k=
\begin{cases}
\lfloor \frac{n}{2^{k-1}} \rfloor /2, \text{ if $\lfloor \frac{n}{2^{k-1}} \rfloor$ is even}. \\
\Big \lfloor \lceil \frac{n}{2^{k-1}} \rceil / 2 \Big \rfloor, \text{ otherwise}.\\
\end{cases}
$$

First we prove the following lemma:
\begin{lemma} \label{lem:LemmaWithProperties}
For any $\lfloor \log{n} \rfloor \ge k >0$, we have that

1. $\sum_{i=1}^{2^k} \delta_k^{i}=n$.

2. For any $1 \le i \le 2^k$, $\delta_k^i \in \{ \lceil \frac {n}{2^k} \rceil \, \lfloor \frac {n}{2^k} \rfloor\}$
(Notice that this means $\alpha_k=\lfloor \frac{n}{2^k}\rfloor$ or $\alpha_k=\lceil \frac{n}{2^k}\rceil$).

\end{lemma}
\begin{proof}
We prove the lemma using induction on $k$.
Base case $k=0$ holds trivially.
For the induction step, assume that Properties 1,2 and 3 hold for $k-1$, we aim to prove that they hold
for $k$ as well.
We have that $\sum_{i=1}^{2^k} \delta_k^{i}=\sum_{i=1}^{2^{k-1}} (\alpha_k+\delta_{k-1}^i-\alpha_k)=\sum_{i=1}^{2^{k-1}} \delta_{k-1}^i=n$.
To prove Property 2 we consider several cases:

\textbf{Case 1}.
$ \frac{n}{2^{k-1}} =2q$, for some integer $q$.

We have that $\alpha_k=q$, and hence for any $1 \le i \le 2^{k-1}$, $\delta_{k-1}^i-\alpha_k=q$.
Since $\lfloor \frac{n}{2^k} \rfloor=q$, Property 2 holds.

\textbf{Case 2}.
$ \frac{n}{2^{k-1}} =2q+1$, for some integer $q$.

We have that $\alpha_k=q$, and hence for any $1 \le i \le 2^{k-1}$, $\delta_{k-1}^i-\alpha_k=q+1$.
Since $\lfloor \frac{n}{2^k} \rfloor=q$ and $\lceil \frac{n}{2^k} \rceil=q+1$, Property 2 holds.

\textbf{Case 3}.
$\frac{n}{2^{k-1}} =2q+\epsilon$, for some integer $q$ and $0 < \epsilon < 1$.

We have that $\lfloor \frac{n}{2^{k-1}} \rfloor=2q$ and $\lceil \frac{n}{2^{k-1}} \rceil=2q+1$.
Additionally, $\alpha_k=q$, and hence for any $1 \le i \le 2^{k-1}$, $(\delta_{k-1}^i-\alpha_k) \in \{q, q+1\}$.
Since $\lfloor \frac{n}{2^k} \rfloor=q$ and $\lceil \frac{n}{2^k} \rceil=q+1$, Property 2 holds.

\textbf{Case 4}.
$ \frac{n}{2^{k-1}} =2q+1+\epsilon$, for some integer $q$ and $0 < \epsilon < 1$.

We have that $\lfloor \frac{n}{2^{k-1}} \rfloor=2q+1$ and $\lceil \frac{n}{2^{k-1}} \rceil=2q+2$.
Additionally, $\alpha_k=q+1$, and hence for any $1 \le i \le 2^{k-1}$, $(\delta_{k-1}^i-\alpha_k) \in \{q, q+1\}$.
Since $\lfloor \frac{n}{2^k} \rfloor=q$ and $\lceil \frac{n}{2^k} \rceil=q+1$, Property 2 holds.

\end{proof}

Next, for $\lfloor \log{n} \rfloor \ge k >0$  we set 
\begin{equation*}
A_i^k=\{i, i+\delta_k^1, i+\delta_k^1+\delta_k^2, ..., i+\sum_{j=1}^{2^k-1} \delta_k^j\}. 
\end{equation*}

It is easy to see that for any $\lfloor \log{n} \rfloor \ge k >0$ and $i$, we have that $|A_{k}^i|=2^k$, $A_{k}^i=A_{k-1}^i \cup A_{k-1}^{i+\alpha_k}$
and $A_{k-1}^i \cap A_{k-1}^{i+\alpha_k}=\emptyset$.
Also notice that 
for any $u \in A_{k-1}^i$, there exists $v \in A_{k-1}^{i+\alpha_k}$, such that $u+\alpha_k=v$ or $v+\alpha_k=u$
(For any $u \in A_{k-1}^{i+\alpha_k}$ there exists $v \in A_{k-1}^{i+\alpha_k}$ with the same property).

Next we prove the lemma which is similar to the lemma for $n=2^m$ case:

\begin{lemma} \label{lem:GeneralUnionOfSets}
For any $1 \le i \le n$ and $\lfloor \log{n} \rfloor \ge k > 0$, we have that
\begin{equation}
2Gap_{A_{k}^i}(t) \le 2 \max_{j \in A_{k}^i}|x_j(t) - x_{j+\alpha_k}(t)| + Gap_{A_{k-1}^{i+\alpha_k}}(t)+
Gap_{A_{k-1}^i}(t).
\end{equation}\end{lemma}

\begin{proof}
Let $u=\argmax_{j \in A_{k-1}^i} x_j(t)$ and let $v =\argmin_{j \in A_{k-1}^i} x_j(t)$.
We consider several cases:

\textbf{Case 1}. $u \in A_{k-1}^{i}$ and $v \in A_{k-1}^{i}$.
Notice that in this case $Gap_{A_{k-1}^i}(t)=Gap_{A_{k}^i}(t)$.
Let $u' \in A_{k-1}^{i+\alpha_k}$ be the vertex such that $u+\alpha_k=u'$ or $u'+\alpha_k=u$ and
let $v' \in A_{k-1}^{i+\alpha_k}$ be the vertex such that $v+\alpha_k=v'$ or $v'+\alpha_k=v$.
We have that
\begin{align*}
    Gap_{A_{k}^i}(t)  &= |x_u(t) - x_v(t)| \\ &\leq |x_{u'}(t) - x_{u}(t)| + |x_{v'}(t) - x_{v}(t)| + |x_{u'}(t) - x_{v'}(t)| \\
    & \leq |x_{u'}(t) - x_{u}(t)| + |x_{v'} - x_{v}(t)| + Gap_{A_{k-1}^{i+\alpha_k}}(t) \\ 
    & \leq  2 \max_{j \in A_{k}^i}|x_j(t) - x_{j+\alpha_k}(t)| + Gap_{A_{k-1}^{i+2^{k-1}}}(t).
\end{align*}
This gives us that
\begin{equation}
2Gap_{A_{k}^i}(t) \le 2 \max_{j \in A_{k}^i}|x_j(t) - x_{j+\alpha_k}(t)| + Gap_{A_{k-1}^{i+\alpha_k}}(t)+
Gap_{A_{k-1}^i}(t).
\end{equation}

\textbf{Case 2}. $u \in A_{k-1}^{i}$ and $v \in A_{k-1}^{i+\alpha_k}$.
Let $u' \in A_{k-1}^{i+\alpha_k}$ be the vertex such that $u+\alpha_k=u'$ or $u'+\alpha_k=u$
and let $v' \in A_{k-1}^{i}$ be the vertex such that $v+\alpha_k=v'$ or $v'+\alpha_k=v$.
We have that:

\begin{align*}
    Gap_{A_{k}^i}(t) = |x_u(t) - x_v(t)| &\leq |x_u(t) - x_{v'}(t)| + |x_{v'}(t) - x_v(t)| \\
    & \leq Gap_{A_{k-1}^{i}}(t) +  \max_{j \in A_{k}^i}(|x_j(t) - x_{j+\alpha_k}(t)|)
\end{align*}
and
\begin{align*}
    Gap_{A_{k}^i}(t) = |x_u(t) - x_v(t)| &\leq |x_u(t) - x_{u'}(t)| + |x_{u'}(t) - x_v(t)| \\
    & \leq Gap_{A_{k-1}^{i+\alpha_k}}(t) +  \max_{j \in A_{k}^i}(|x_j(t) - x_{j+\alpha_k}(t)|)
\end{align*}
Hence, we again get that
\begin{equation}
2Gap_{A_{k}^i}(t) \le 2 \max_{j \in A_{k}^i}|x_j(t) - x_{j+\alpha_k}(t)| + Gap_{A_{k-1}^{i+\alpha_k}}(t)+
Gap_{A_{k-1}^i}(t).
\end{equation}

\textbf{Case 3}. $u \in A_{k-1}^{i+\alpha_k}$ and $v \in A_{k-1}^{i+\alpha_k}$, is similar to Case 1.

\textbf{Case 4}. $v \in A_{k-1}^{i}$ and $u \in A_{k-1}^{i+\alpha_k}$, is similar to Case 2.
\end{proof}

Next, we upper bound $\sum_{i=1}^n \max_{j \in A_{k}^i}|x_j(t) - x_{j+\alpha_k}(t)|.$

\begin{lemma} \label{lem:GeneralSumOfEdges}
\begin{equation}
\sum_{i=1}^{n} \max_{j \in A_{k}^i}|x_j(t) - x_{j+\alpha_k}(t)| \le 
\Big\lceil \frac{n}{\lfloor \frac{n}{2^k} \rfloor} \Big\rceil \sqrt{\lfloor \frac{n}{2^k}\rfloor}  \sqrt{\phi_{\alpha_k}(t)}
\end{equation}
\end{lemma}
\begin{proof}
Notice that for any $1\le u \le n$ and sets $A_k^u, A_k^{u+1}, ..., A_k^{u+\lfloor \frac{n}{2^k} \rfloor-1}$
are disjoint, because for any $1 \le j \le 2^k$, $\delta_k^j \ge \lfloor \frac{n}{2^k} \rfloor$
(This means that for any $1 \le i \le n$, distances between consecutive vertices in $A_k^i$ are at least $\lfloor \frac{n}{2^k} \rfloor$).
Using this fact and Cauchy-Schwarz inequality we get that
\begin{align*}
&\sum_{i=u}^{u+\lfloor \frac{n}{2^k} \rfloor-1} \max_{j \in A_{k}^i}|x_j(t) - x_{j+\alpha_k}(t)|  \\&\le 
\sqrt{\lfloor \frac{n}{2^k}\rfloor} \sqrt{\sum_{i=u}^{u+\lfloor \frac{n}{2^k} \rfloor-1}
\max_{j \in A_{k}^{i}}|x_j(t) - x_{j+\alpha_k}(t)|^2} \\&\le
\sqrt{\lfloor \frac{n}{2^k}\rfloor} \sqrt{\sum_{j=1}^n |x_j(t) - x_{j+\alpha_k}(t)|^2}=\sqrt{\lfloor \frac{n}{2^k}\rfloor}  \sqrt{\phi_{\alpha_k}(t)}
\end{align*}
Since the above inequality holds for any $u$ we can write that:
\begin{equation}
\sum_{i=1}^{n} \max_{j \in A_{k}^i}|x_j(t) - x_{j+\alpha_k}(t)| \le \Big\lceil \frac{n}{\lfloor \frac{n}{2^k} \rfloor} \Big\rceil \sqrt{\lfloor \frac{n}{2^k}\rfloor}  \sqrt{\phi_{\alpha_k}(t)}
\end{equation}

\end{proof}
With the above lemmas in place,
we are ready to prove Theorem \ref{thm:GapUpperBound} for general $n$.

From Lemma \ref{lem:GeneralUnionOfSets} we have that 
\begin{align*}
\sum_{i=1}^n 2Gap_{A_{k}^i(t)} &\le \sum_{i=1}^n Gap_{A_{k-1}^i}(t)+\sum_{i=1}^n Gap_{A_{k-1}^{i+\alpha_k}}(t)\\& \quad \quad \quad +
\sum_{i=1}^n 2 \max_{j \in A_{k}^i}|x_j(t) - x_{j+\alpha_k}(t)| \\&=
2\sum_{i=1}^n Gap_{A_{k-1}^i}(t)+2\sum_{i=1}^n \max_{j \in A_{k}^i}|x_j(t) - x_{j+\alpha_k}(t)|.
\end{align*}
After dividing the above inequality by 2 and applying Lemma \ref{lem:GeneralSumOfEdges}:
we get that:
\begin{align*}
\sum_{i=1}^n Gap_{A_{k}^i}(t) \le 
\sum_{i=1}^n Gap_{A_{k-1}^i}(t)+\Big\lceil \frac{n}{\lfloor \frac{n}{2^k} \rfloor} \Big\rceil \sqrt{\lfloor \frac{n}{2^k}\rfloor}  \sqrt{\phi_{\alpha_k}(t)}.
\end{align*}
Notice that for any $i$, $Gap_0^i(t)=0$.
Hence, we get that 
\begin{align*}
\sum_{i=1}^n Gap_{A_{\lfloor \log{n} \rfloor}^i} (t) \le 
\sum_{k=1}^{\lfloor \log{n} \rfloor} \Big\lceil \frac{n}{\lfloor \frac{n}{2^k} \rfloor} \Big\rceil \sqrt{\lfloor \frac{n}{2^k}\rfloor}  \sqrt{\phi_{\alpha_k}(t)}
.
\end{align*}
Let $i'=\argmin_i{Gap_{A_{\lfloor \log{n} \rfloor}^i} (t)}$.
Notice that consecutive vertices in $A_{\lfloor \log{n} \rfloor}^{i'}$ are 1 or 2 edges apart,
hence for any $1 \le i \le n$, either $i \in A_{\lfloor \log{n} \rfloor}^{i'}$ or $i+1 \in A_{\lfloor \log{n} \rfloor}^{i'}$.
This gives us that 
\begin{align*}
Gap(t) &\le Gap_{A_{\lfloor \log{n} \rfloor}^{i'}} (t)+\max_{i} |x_i(t)-x_{i+1}(t)| \\
&= Gap_{A_{\lfloor \log{n} \rfloor}^{i'}} (t)+\sqrt{\max_{i} |x_i(t)-x_{i+1}(t)|^2} \le
Gap_{A_{\lfloor \log{n} \rfloor}^{i'}} (t)+\sqrt{\phi_1(t)}.
\end{align*}
By combining the above two inequalities we get that
\begin{align*}
n Gap(t) &\le n Gap_{A_{\lfloor \log{n} \rfloor}^{i'}} (t)+n\sqrt{\phi_1(t)} \le  
\sum_{i=1}^n Gap_{A_{\lfloor \log{n} \rfloor}^i} (t)+n\sqrt{\phi_1(t)} \\ &\le
\sum_{k=1}^{\lfloor \log{n} \rfloor}  \Big\lceil \frac{n}{\lfloor \frac{n}{2^k} \rfloor} \Big\rceil \sqrt{\lfloor \frac{n}{2^k}\rfloor}  \sqrt{\phi_{\alpha_k}(t)}+n\sqrt{\phi_1(t)}.
\end{align*}

Next, we apply Jensen's inequality and Lemma \ref{thm:UpperBoundOnPotentials} (We are going to use
a looser upper bound: $\E[\phi_i(t)] \le i(n-i)-1 \le in)$
\begin{align*}
n \E[Gap(t)]  &\le
n \E\sqrt{[\phi_{1}(t)]}+\sum_{k=1}^{\lfloor \log{n} \rfloor} \Big\lceil \frac{n}{\lfloor \frac{n}{2^k} \rfloor} \Big\rceil \sqrt{\lfloor \frac{n}{2^k}\rfloor}  \E\sqrt{\phi_{\alpha_k}(t)}
\\ &\le
n\sqrt{\E[\phi_{1}(t)]}+\sum_{k=1}^{\lfloor \log{n} \rfloor} \Big\lceil \frac{n}{\lfloor \frac{n}{2^k} \rfloor} \Big\rceil \sqrt{\lfloor \frac{n}{2^k}\rfloor}  \sqrt{\E[\phi_{\alpha_k}(t)]}
\\ &\le
n\sqrt{n}+\sum_{k=1}^{\lfloor \log{n} \rfloor} \Bigg( \Big\lceil \frac{n}{\lfloor \frac{n}{2^k} \rfloor} \Big\rceil \sqrt{\lfloor \frac{n}{2^k}\rfloor}  \sqrt{\alpha_k n}
=O(n\sqrt{n}\log{n}).
\end{align*}

This completes the proof.

\section{Harary Graph, upper bound on the gap}
In this section we show that our approach can be used to upper bound the gap for the Harary graph with $n$ vertices:
That is each vertex $i$ is connected with edge to vertices $i-1$, $i+1$ (called cycle edges), $i-2$ and $i+2$ (called
extra edges).
As before the operation consists of
picking an edge u.a.r and doing increment and averaging (For the simplicity we assume that increments have unit weights, the result can be extended to the random weights, in the similar fashion to the cycle case).
After careful calculations which mimic the calculations for the cycle case, but 
by taking extra edges of the Harary graph into the account, we can derive the following equations for 
the hop potentials (Hops are counted over the cycle edges):

\[
\begin{cases}
\mathbb{E}[\phi_1(t+1)] =  \frac{n-2}{n}\E[\phi_1(t)]+\frac{3}{4}+\frac{\E[\phi_1(t)]}{4n}+
\frac{\E[\phi_3(t)]}{2n}\\
\mathbb{E}[\phi_2(t+1)] =  \frac{n-2}{n}\E[\phi_2(t)]+\frac{3}{4}-\frac{\E[\phi_2(t)]}{4n}+
\frac{\E[\phi_3(t)]}{2n}+\frac{\E[\phi_4(t)]}{2n}\\
\dots \\
\mathbb{E}[\phi_{k}(t+1)] =  \frac{n-2}{n}\E[\phi_k(t)]+1-\frac{\E[\phi_1(t)]}{2n}-\frac{\E[\phi_2(t)
]}{2n}+
\frac{\E[\phi_{k-2}(t)]}{2n}+\frac{\E[\phi_{k-1}(t)]}{2n}\\ \quad \quad \quad \quad
\quad \quad \quad \quad + \frac{\E[\phi_{k+1}(t)]}{2n}+
\frac{\E[\phi_{k+2}(t)]}{2n} \\
\end{cases}
\]
Recall that for the cycle potential $\phi_k(t+1)$ depends on the potentials $\phi_k(t)$, $\phi_{k+1}(t)$, 
$\phi_{k-1}(t)$ and $\phi_1(t)$. In the case of Harary graph $\phi_k(t+1)$ depends on the potentials $\phi_{k+2}(t)$, $\phi_{k+1}(t)$, $\phi_k(t)$, $\phi_{k-1}(t)$, $\phi_{k-2}(t)$, $\phi_1(t)$ and $\phi_2(t)$.
The reason is that we are to able to perform load balancing operation on two hop neighbours. 
Similarly, if we have a graph where each vertex is connected with all vertices which are at hop distance at most $\ell$ (This is also Harary graph, but with different parameter), then
$\phi_k(t+1)$ will depend on $\phi_{k+\ell}(t)$, $\dots$, $\phi_{k+1}(t)$, $\phi_{k}(t)$, $\phi_{k-1}(t)$, $\dots$, $\phi_{k-\ell}(t)$, and $\phi_1(t)$, $\phi_2(t)$, $\dots$, $\phi_{\ell}(t)$.
Next step is to find the stationary points for the hop potentials. That is:
the values which stay the same after we apply step given by the above equations.
As before these values will be used as the upper bounds for the expected values of potentials.
In the case of Harary Graph (with 2 hop edges), we get that for every $t$ and $k$:
\begin{equation}
\E[\phi_k(t)] \le \frac{2}{5}k(n-k)+\alpha.
\end{equation}
Here extra term $\alpha$ has a closed form, which we omit and instead concentrate on the property that it is upper bounded by $2n$.
Observe that since Harary graph contains cycle and we defined our hop potentials
based on the hop counts of that cycle, we can upper bound the gap by using:
\begin{equation}
\E[Gap(t)] \le \sum_{k=0}^m \frac{1}{\sqrt{2^{k-1}}}\sqrt{\phi_{2^{k-1}}(t)} \le O(\sqrt{n}\log{n}),
\end{equation}
Notice  that if $k \ge 40$,
 the upper bound for $\E[\phi_k(t)]$ is two times smaller than the upper bound for the cycle case.
Hence, we can use this to slightly improve the constant hidden by big $O$ notation in the upper bound.
 \end{document}